\documentclass{llncs}

\usepackage{amsmath,amssymb,amsfonts}
\usepackage{caption}
\usepackage{subcaption}
\usepackage{csquotes}
\usepackage{float}
\usepackage{marvosym}
\usepackage{microtype}
\usepackage{multirow}

\usepackage[sort]{cleveref}

\usepackage{tikz}

\usetikzlibrary{arrows,snakes}

\crefformat{page}{#2page~#1#3}%
\Crefformat{page}{#2Page~#1#3}%
\crefformat{equation}{#2(#1)#3}%
\Crefformat{equation}{#2(#1)#3}%
\crefformat{figure}{#2Figure~#1#3}%
\Crefformat{figure}{#2Figure~#1#3}%
\crefformat{section}{#2Section~#1#3}
\Crefformat{section}{#2Section~#1#3}
\crefformat{appendix}{#2Appendix~#1#3}
\Crefformat{appendix}{#2Appendix~#1#3}
\crefformat{chapter}{#2Chapter~#1#3}
\Crefformat{chapter}{#2Chapter~#1#3}
\crefformat{chapter*}{#2Chapter~#1#3}
\Crefformat{chapter*}{#2Chapter~#1#3}
\crefformat{part}{#2Part~#1#3}
\Crefformat{part}{#2Part~#1#3}
\crefformat{enumi}{#2(#1)#3}
\Crefformat{enumi}{#2(#1)#3}

\newcommand*{\defeq}{\mathrel{\vcenter{\baselineskip0.5ex \lineskiplimit0pt
                     \hbox{\scriptsize.}\hbox{\scriptsize.}}}%
                     =}

\newcommand{\googlevspace}{\vspace{-6mm}}

\newcommand{\shortd}{\ensuremath{\check{d}}}

\newcommand{\normal}{\mathcal{N}}
\newcommand{\expected}[1]{\mathbb{E}[#1]}
\newcommand{\N}{\mathbb{N}}
\newcommand{\Z}{\mathbb{Z}}

\DeclareMathOperator*{\median}{median}

\graphicspath{{../imgs/}{imgs/}}

\title{TRIX: Low-Skew Pulse Propagation for Fault-Tolerant Hardware}

\author{
  Ben Wiederhake\thanks{Max Planck Institute for
      Informatics, Saarland Informatics Campus and Saarbrücken
      Graduate School of Computer Science, Germany; \texttt{bwiederh@mpi-inf.mpg.de}}
  \and
  Christoph Lenzen\thanks{Max Planck Institute for
      Informatics, Saarland Informatics Campus; \texttt{clenzen@mpi-inf.mpg.de}}
}

\begin{document}

\maketitle

\begin{abstract}
\noindent
The vast majority of hardware architectures use a carefully timed reference signal to clock their computational logic.
However, standard distribution solutions are not fault-tolerant.
In this work, we present a simple grid structure as a more reliable clock propagation method and study it by means of simulation experiments.
Fault-tolerance is achieved by forwarding clock pulses on arrival of the second of three incoming signals from the previous layer.

A key question is how well neighboring grid nodes are synchronized, even without faults.
Analyzing the clock skew under typical-case conditions is highly challenging.
Because the forwarding mechanism involves taking the median, standard
probabilistic tools fail, even when modeling link delays just by unbiased coin flips.

Our statistical approach provides substantial evidence that this system performs surprisingly well.
Specifically, in an \enquote{infinitely wide} grid of height~$H$, the delay at a pre-selected node exhibits a standard deviation of $O(H^{1/4})$
($\approx 2.7$ link delay uncertainties for $H=2000$) and skew between adjacent nodes of $o(\log \log H)$
($\approx 0.77$ link delay uncertainties for $H=2000$).
We conclude that the proposed system is a very promising clock distribution method.
This leads to the open problem of a stochastic explanation of the tight concentration of delays and skews.
More generally, we believe that understanding our very simple abstraction of the
system is of mathematical interest in its own right.
\end{abstract}


\section{Introduction}\label{sec:intro}

When designing high reliability systems, any critical subsystem susceptible to
failure must exhibit sufficient redundancy.
Traditionally, clocking of synchronous systems is performed by clock trees or
other structures that cannot sustain faulty components~\cite{xanthopoulos09}.
This imposes limits on scalability on the size of clock domains; for instance,
in multi-processor systems typically no or only very loose synchronization is
maintained between different processors~\cite{culler98,patterson1990}.
Arguably, this suggests that fault-tolerant clocking methods that are
competitive -- or even better -- in terms of synchronization quality and other
parameters (ease of layouting, amount of circuitry, energy consumption, etc.)
would be instrumental in the design of larger synchronous systems.


To the best of our knowledge, at least until~20, or perhaps
even~10 years ago, there is virtually no work on fault-tolerance of clocking schemes
beyond production from the hardware community; due to the size and degree of
miniaturization of systems at the time, clock trees and their derivatives were
simply sufficiently reliable in practice.
That this has changed is best illustrated by an upsurge of interest in single
event upsets of the clocking subsystem in the last
decade~\cite{abouzeid15,chipana14,chipana12,chipana11,gujja15,malherbe16,wang16,wissel09}.
However, with ever larger systems and smaller components in place, achieving
acceptable trade-offs between reliability, synchronization quality, and energy
consumption requires to go beyond these techniques.

On the other hand, there is a significant body of work on fault-tolerant
synchronization from the area of distributed systems.
Classics are the Srikanth-Toueg~\cite{srikanth87optimal} and
Lynch-Welch~\cite{welch88} algorithms, which maintain synchronization even in
face of a large minority (strictly less than one third) of \emph{Byzantine}
faulty nodes.\footnote{Distributed systems are typically modeled by network
graphs, where the nodes are the computational devices and edges represent
communication links. A Byzantine faulty node may deviate from the protocol in an
arbitrary fashion, i.e., it models worst-case faults and/or malicious attacks on
the system.} Going beyond this already very strong fault model, a line of
works~\cite{daliot03self,dolev07bounded,dolev14fatal,dolev04clock,lenzen19easyJ}
additionally consider \emph{self-stabilization,} the ability of a system to
recover from an unbounded number of transient faults.
The goal is for the system to stabilize, i.e., recover nominal operation, after
transient faults have ceased.
Note that the combination makes for a very challenging setting and results in
extremely robust systems: even if some nodes remain faulty, the system will
recover from transient faults, which is equivalent to recover synchronization
when starting from an arbitrary state despite interference from Byzantine faulty
nodes.

While these fault-tolerance properties are highly desirable, unsurprisingly they
also come at a high price.
All of the aforementioned works assume a fully connected system, i.e., direct
connections between each pair of nodes.
Due to the strong requirements, it is not hard to see that this is essentially
necessary~\cite{dolev86impossibility}:
in order to ensure that each non-faulty node can synchronize to the majority of
correct nodes, its degree must exceed the number of faults, or it might become
effectively disconnected.
In fact, it actually must have more correct neighbors than faulty ones, or a
faulty majority of neighbors might falsely appear to provide the correct time.
Note that emulating full connectivity using a crossbar or some other sparser
network topology defeats the purpose, as the system then will be brought down by
a much smaller number of faults in the communication infrastructure connecting
the nodes.
Accordingly, asking for such extreme robustness must result in solutions that do
not scale well.

A suitable relaxation of requirements is proposed
in~\cite{DBLP:journals/jcss/DolevFLPS16}.
Instead of assuming that Byzantine faults are also \emph{distributed} across the
system in a worst-case fashion, the authors of this work require that faults are
``spread out.'' More specifically, they propose a grid-like network they call
HEX, through which a clock signal can be reliably distributed, so long as for
each node at most one of its four in-neighbors is faulty.
Note that for the purposes of this paper, we assume that the problem of fault-tolerant clock signal
\emph{generation} has already been sufficiently addressed (e.g. using~\cite{dolev14fatal}), but the signal still needs to be
\emph{distributed}.
Provided that nodes fail independently, this means that the probability of
failure of individual nodes that can likely be sustained becomes roughly~$1/\sqrt{n}$,
where~$n$ is the total number of nodes; this is to be contrasted
with a system without fault-tolerance, in which components must fail with
probability at most roughly~$1/n$.
The authors also show how to make HEX self-stabilizing.
Unfortunately, however, the approach has poor synchronization performance even
in face of faults obeying the constraint of at most one fault in each
in-neighborhood.
While it is guaranteed that the clock signal propagates through the grid, nodes
that fail to propagate the clock signal cause a ``detour'' resulting in a clock
skew between neighbors of at least one maximum end-to-end communication delay
$d$.\footnote{$d$ includes the wire delay as well as the time required for local
computations. As the grid is highly uniform and links connect close-by nodes,
the reader should expect this value to be roughly the same for all links.} This
is much larger than the \emph{uncertainty} $u$ in the end-to-end delay:
As clocking systems can be (and are) engineered for this purpose, the end-to-end
delay will vary between $d-u$ and~$d$ for some $u\ll d$.
To put this into perspective, in a typical system~$u$ will be a fraction of a
clock cycle, while~$d$ may easily be half of a clock cycle or more.

This is inherent to the structure of the HEX grid, see \Cref{fig:hex_worst}.
It seeks to propagate the clock signal from layer to layer, where each node has
two in-neighbors on the preceding and two in-neighbors on their own layer.
Because the possibility of a fault requires nodes to wait for at least two
neighbors indicating a clock pulse before doing so themselves, a faulty node
refusing to send any signal implies that its two out-neighbors on the next layer
need to wait for at least one signal from their own layer.
This adds at least one hop to the path along which the signal is propagated,
causing an additional delay of at least roughly~$d$.

\paragraph*{Our Contribution}


\begin{figure}
  \centering
  \begin{minipage}{.56\textwidth}
    \centering
\begin{tikzpicture}[xscale=1.1547,yscale=1,
hnode/.style={draw,circle,minimum size=0.8cm,fill=white},
linkdead/.style={->,dotted},
linkpass/.style={->,thick,dotted},
linkact/.style={->,thick}
]

\draw[<->]  (0,4) -- (0,0) -- (5,0);
\node[anchor=north] at (2.5,0) {column};
\node[rotate=90,anchor=south] at (0,2) {layer};


\node[hnode] (a1) at (1,1) {3d};
\node[hnode] (a2) at (2,1) {3d};
\node[hnode] (a3) at (3,1) {3d};
\node[hnode] (a4) at (4,1) {3d};

\node[hnode] (b1) at (0.5,2) {4d};
\node[hnode] (b2) at (1.5,2) {4d};
\node[hnode] (b3) at (2.5,2) {};
\node at (b3) {\Huge\Lightning};
\node[hnode] (b4) at (3.5,2) {4d};
\node[hnode] (b5) at (4.5,2) {4d};

\node[hnode] (c1) at (1,3) {5d};
\node[hnode] (c2) at (2,3) {6d};
\node[hnode] (c3) at (3,3) {6d};
\node[hnode] (c4) at (4,3) {5d};


\draw[linkact] (a1) -- (b1);
\draw[linkact] (a1) -- (b2);
\draw[linkact] (a2) -- (b2);
\draw[linkact] (a2) -- (b3);
\draw[linkact] (a3) -- (b3);
\draw[linkact] (a3) -- (b4);
\draw[linkact] (a4) -- (b4);
\draw[linkact] (a4) -- (b5);

\draw[linkact] (b1) -- (c1);
\draw[linkact] (b2) -- (c1);
\draw[linkact] (b2) -- (c2);
\draw[linkdead] (b3) -- (c2);
\draw[linkdead] (b3) -- (c3);
\draw[linkact] (b4) -- (c3);
\draw[linkact] (b4) -- (c4);
\draw[linkact] (b5) -- (c4);

\draw[linkpass] (a1) to[out=-20,in=-160] (a2); \draw[linkpass] (a2) to[out=160,in=20] (a1);
\draw[linkpass] (a2) to[out=-20,in=-160] (a3); \draw[linkpass] (a3) to[out=160,in=20] (a2);
\draw[linkpass] (a3) to[out=-20,in=-160] (a4); \draw[linkpass] (a4) to[out=160,in=20] (a3);

\draw[linkpass] (b1) to[out=-20,in=-160] (b2); \draw[linkpass] (b2) to[out=160,in=20] (b1);
\draw[linkpass] (b2) to[out=-20,in=-160] (b3); \draw[linkdead] (b3) to[out=160,in=20] (b2);
\draw[linkdead] (b3) to[out=-20,in=-160] (b4); \draw[linkpass] (b4) to[out=160,in=20] (b3);
\draw[linkpass] (b4) to[out=-20,in=-160] (b5); \draw[linkpass] (b5) to[out=160,in=20] (b4);

\draw[linkact] (c1) to[out=-20,in=-160] (c2); \draw[linkpass] (c2) to[out=160,in=20] (c1);
\draw[linkpass] (c2) to[out=-20,in=-160] (c3); \draw[linkpass] (c3) to[out=160,in=20] (c2);
\draw[linkpass] (c3) to[out=-20,in=-160] (c4); \draw[linkact] (c4) to[out=160,in=20] (c3);

\end{tikzpicture}
    \captionof{figure}{A crashing node in a HEX grid causes a large skew between neighbors
    in the same layer, even with all links having exactly the same delay. Thick
    links cause nodes to pulse, dotted links mean that the transmitted pulse
    was too late to be considered, and faint dotted links do not transmit a pulse.}
    \label{fig:hex_worst}
  \end{minipage}%
  \hfill{}%
  \begin{minipage}{.4\textwidth}
    \centering
    \begin{tikzpicture}[scale=1.304,
hnode/.style={draw,circle,minimum size=0.8cm,fill=white},
link/.style={->,thick}
]

\draw[<->]  (0.5,3.5) -- (0.5,0.5) -- (3.5,0.5);
\node[anchor=north] at (2,0.5) {column};
\node[rotate=90,anchor=south] at (0.5,2) {layer};

\foreach \x/\l in {1/a, 2/b, 3/c} {
    \foreach \y in {1, 2,3} {
        \node[hnode] (\l\y) at (\x,\y) {};
    }
}

\draw[link] (a1) -- (b2);
\draw[link] (b1) -- (b2);
\draw[link] (c1) -- (b2);

\draw[link] (b2) -- (a3);
\draw[link] (b2) -- (b3);
\draw[link] (b2) -- (c3);

\path (0.5,0.5) -- (3.75,0.5); 

\end{tikzpicture}
    \captionof{figure}{The basic topology of TRIX. Only wires incident to the central node are shown;
    the same pattern is applied at other nodes.}
    \label{fig:trix_simple}
  \end{minipage}
\end{figure}

\begin{figure}
  \centering
  \begin{minipage}{.48\textwidth}
    \centering
\begin{tikzpicture}[scale=1.19,
hnode/.style={draw,circle,minimum size=0.8cm,fill=white},
dead/.style={->,dotted,->},
fast/.style={->,thick},
slow/.style={->,decoration={snake,amplitude=.4mm,segment length=1mm,pre length=.5mm,post length=.5mm},decorate}
]

\draw[<->]  (0.5,4.5) -- (0.5,0.5) -- (4.5,0.5);
\node[anchor=north] at (2.5,0.5) {column};
\node[rotate=90,anchor=south] at (0.5,2.5) {layer};

\foreach \x/\l in {1/a, 2/b, 3/c, 4/d} {
    \foreach \y\u in {1/0, 2/1, 3/2, 4/3} {
        \node[hnode] (\l\y) at (\x,\y) {$\u{}d$};
    }
}
\node[hnode] (b2) at (2,2) {};
\node (b2L) at (2,2) {\Huge\Lightning};


\foreach \x in {a, c, d} {
    \foreach \y/\v in {1/2, 2/3, 3/4} {
        \draw[fast] (\x\y) -- (\x\v);
    }
}
\foreach \x in {b} {
    \foreach \y/\v in {1/2, 3/4} {
        \draw[fast] (\x\y) -- (\x\v);
    }
}

\foreach \x/\u in {a/b, c/d} {
    \foreach \y/\v in {1/2, 2/3, 3/4} {
        \draw[fast] (\x\y) -- (\u\v);
    }
}
\foreach \x/\u in {b/c} {
    \foreach \y/\v in {1/2, 3/4} {
        \draw[fast] (\x\y) -- (\u\v);
    }
}

\foreach \x/\u in {c/b, d/c} {
    \foreach \y/\v in {1/2, 2/3, 3/4} {
        \draw[fast] (\x\y) -- (\u\v);
    }
}
\foreach \x/\u in {b/a} {
    \foreach \y/\v in {1/2, 3/4} {
        \draw[fast] (\x\y) -- (\u\v);
    }
}

\draw[dead] (b2) -- (a3);
\draw[dead] (b2) -- (b3);
\draw[dead] (b2) -- (c3);

\path (0.5,0.5) -- (4.75,0.5); 

\end{tikzpicture}
    \captionof{figure}{A crashing node in a TRIX grid causes no significant skew,
    compared to \Cref{fig:hex_worst}. In fact, in absence of uncertainty,
    isolated crashes can be ignored entirely.}
    \label{fig:trix_singlefault}
  \end{minipage}%
  \hfill{}%
  \begin{minipage}{.48\textwidth}
    \centering
\begin{tikzpicture}[scale=1.19,
hnode/.style={draw,circle,minimum size=0.8cm,fill=white},
fast/.style={->,thick},
slow/.style={->,decoration={snake,amplitude=.4mm,segment length=1mm,pre length=.5mm,post length=.5mm},decorate}
]

\draw[<->]  (0.5,4.5) -- (0.5,0.5) -- (4.5,0.5);
\node[anchor=north] at (2.5,0.5) {column};
\node[rotate=90,anchor=south] at (0.5,2.5) {layer};

\foreach \x/\l in {1/a, 2/b} {
    \foreach \y\u in {1/0, 2/1, 3/2, 4/3} {
        \node[hnode] (\l\y) at (\x,\y) {$\u{}d$};
    }
}
\foreach \x/\l in {3/c, 4/d} {
    \foreach \y\u in {1/0, 2/1, 3/2, 4/3} {
        \node[hnode] (\l\y) at (\x,\y) {$\u{}\shortd$};
    }
}


\foreach \x in {a, b} {
    \foreach \y/\v in {1/2, 2/3, 3/4} {
        \draw[slow] (\x\y) -- (\x\v);
    }
}
\foreach \x in {c, d} {
    \foreach \y/\v in {1/2, 2/3, 3/4} {
        \draw[fast] (\x\y) -- (\x\v);
    }
}

\foreach \x/\u in {a/b, b/c} {
    \foreach \y/\v in {1/2, 2/3, 3/4} {
        \draw[slow] (\x\y) -- (\u\v);
    }
}
\foreach \x/\u in {c/d} {
    \foreach \y/\v in {1/2, 2/3, 3/4} {
        \draw[fast] (\x\y) -- (\u\v);
    }
}

\foreach \x/\u in {b/a} {
    \foreach \y/\v in {1/2, 2/3, 3/4} {
        \draw[slow] (\x\y) -- (\u\v);
    }
}
\foreach \x/\u in {c/b, d/c} {
    \foreach \y/\v in {1/2, 2/3, 3/4} {
        \draw[fast] (\x\y) -- (\u\v);
    }
}

\path (0.5,0.5) -- (4.75,0.5); 

\end{tikzpicture}
    \captionof{figure}{Worst-case assignment of wire delays causing large skew for TRIX.
    Squiggly lines indicate slow wires, straight ones fast wires.
    The symbol $\shortd$ stands for $d-u$.}
    \label{fig:trix_worst}
  \end{minipage}
\end{figure}

We propose a novel clock distribution topology that overcomes the above
shortcoming of HEX.
As in our topology nodes have in- and out-degrees of~$3$ and it is inspired by
HEX, we refer to it as TRIX; see \Cref{fig:trix_simple} for an
illustration of the grid structure.
Similar to HEX, the clock signal is propagated through layers, but for each node,
all of its three in-neighbors are on the preceding layer.
This avoids the pitfall of faulty nodes significantly slowing down the
propagation of the signal. If at most one in-neighbor is faulty,
each node still has two correct in-neighbors on the preceding layer, as demonstrated in \Cref{fig:trix_singlefault}.
Hence, we can now focus on fault-free executions, because single isolated
faults only introduce an additional uncertainty of at most $u \ll d$.
Predictions in this model are therefore
still meaningful for systems with rare and non-malicious faults.

The TRIX topology is acyclic, which conveniently means that
self-stabilization is trivial to achieve, as any incorrect state is \enquote{flushed
out} from the system.

Despite its apparent attractiveness and even greater simplicity, we note that
this choice of topology should not be obvious.
The fact that nodes do not check in with their neighbors on the same layer
implies that the worst-case clock skew between neighbors grows as~$u H$, where
$H$ is the number of layers and (for the sake of simplicity) we assume that the
skew on the first layer (which can be seen as the ``clock input'') is~$0$, see
\Cref{fig:trix_worst}.
However, reaching the skew of~$d$ between neighbors on the same layer, which is
necessary to give purpose to any link between them, takes many layers, at least $d/u\gg 1$
many. This is in contrast to HEX, where the worst-case skew is bounded, but more easily attained.
When not assuming that delays are chosen in a worst-case manner, our statistical
experiments show it to likely take much longer before this threshold is reached.
Accordingly, in TRIX there is no need for links within the same layer for any
practical number of layers, resulting in the advantage of smaller in- and
out-degrees compared to HEX.

The main focus of this work lies on statistical experiments with the goal of
estimating the performance of a TRIX grid as clock distribution method.
Note that this is largely\footnote{Skews over longer distances are relevant for
long-range communication, but have longer communication delays and respective
uncertainties. This entails larger buffers even in absence of clock skew. We
briefly show that the TRIX grid appears to behave well also in this regard.}
dominated by the skew between adjacent nodes in the grid, as these will drive
circuitry that needs to communicate.
While the worst-case behavior is easy to understand, it originates from a very
unlikely configuration, where one side of the grid is entirely slow and the
other is fast, see \Cref{fig:trix_worst}.
In contrast, correlated but gradual changes will also result in spreading out
clock skews -- and any change that affects an entire region in the same way will
not affect local timing differences at all.
This motivates to study the extreme case of independent noise on each link in
the TRIX grid.
Choosing a simple abstraction, we study the random process in which each link is
assigned either delay~$0$ or delay~$1$ by an independent, unbiased coin flip.
Moreover, we assume ``perfect'' input, i.e., each node on the initial layer
signals a clock pulse at time~$0$, and that the grid is infinitely
wide.\footnote{Experiments with grids of bounded width suggest that reducing
width only helps, while our goal here is to study the asymptotic behavior for
large systems.} By induction over the layers, each node is then assigned the
second largest value out of the three integers obtained by adding the respective
link delay to the pulse time of each of its in-neighbors.
We argue that this simplistic abstraction captures the essence of (independent)
noise on the channels.

Due to the lack of applicable concentration bounds for such processes, we study
this random process by extensive numerical experiments.
Our results provide evidence that TRIX behaves surprisingly well in several regards,
exhibiting better concentration than one might expect.
First and foremost, the skew between neighbors appears to grow extremely slowly
with the number of layers~$H$.
Even for~$2000$ layers, we never observed larger differences than~$7$ between
neighbors on the same layer.
Plotting the standard deviation of the respective distribution as a function of
the layer, the experiments show a growth that is slower than \emph{doubly
logarithmic}, i.e.~$\log \log H$. Second, for a fixed layer, the respective skew distribution
exhibits an exponential tail falling as roughly $e^{-\lambda |x|}$ for $\lambda
\approx 2.9$.
Third, the distribution of pulsing times as a function of the layer (i.e., when
a node pulses, not the difference to its neighbors) is also concentrated around
its mean (which can easily be shown to be~$H/2$), where the standard deviation
grows roughly as~$H^{1/4}$.

To support that these results are not simply artifacts of the simulation, i.e.,
that we sampled sufficiently often, we make use of the
Dvoretzky-Kiefer-Wolfowitz (DKW) inequality~\cite{dvoretzky1956} to show that the
underlying ground truth is very close to the observed distributions. In
addition, to obtain tighter error bounds for our asymptotic analysis of
standard deviations as function of~$H$, we leverage the Chernoff bound (as
stated in~\cite{mitzenmacher05}) on individual values.
We reach a high confidence that the qualitative assertion that the TRIX
distribution exhibits surprisingly good concentration is well-founded.
We conclude that the simple mechanism underlying the proposed clock distribution
mechanism results in a fundamentally different behavior than existing clock
distribution methods or naive averaging schemes.


\section{Preliminaries}\label{sec:model}
\label{subsec:model}
\paragraph*{Model}

We model TRIX in an abstract way that is amenable to very efficient simulation
in software.
In this section, we introduce this model and discuss the assumptions and the
resulting restrictions in detail.

The network topology is a grid of height~$H$ and width~$W$.
For finite~$W$, the left- and rightmost column would be connected, resulting in
a cylinder.
To simplify, we choose $W = \infty$, because we aim to focus on the behavior in
large systems.
Note that a finite width will work in our favor, as it adds additional
constraints on how skews can evolve over layers; in the extreme case of~$W=3$,
in absence of faults skews could never become larger than~$1$.
We refer to the grid nodes by integer coordinates $(x,y)$, where $x\in \Z$ and
$y\in \N_0$.
Layer $0\leq \ell \leq H$ consists of the nodes $(x,\ell)$, $x\in \Z$.

Layer~$0$ is special in that its nodes represent the clock source;
they always \emph{pulse} at time~$0$.
Again, there are implicit simplifications here.
First, synchronizing the layer~$0$ nodes requires a suitable solution -- ideally
also fault-tolerant -- and cannot be done perfectly.
However, our main goal here is to understand the properties of the clock
distribution grid, so the initial skew is relevant only insofar as it affects
the distribution.
Some indicative simulations demonstrate that the grid can counteract ``bad''
inputs to some extent, but some configurations do not allow for this in a few
layers, e.g.\ having all nodes with negative~$x$ coordinates pulse much earlier
than those with positive ones.
Put simply, this would be a case of ``garbage-in garbage-out,'' which is not the
focus of this study.
A more subtle point is the unrealistic assumption that all layers of the grid have the same width.
If the layer~$0$ nodes are to be well-synchronized, they ought to be physically
close; arranging them in a wide line is a poor choice.
Accordingly, arranging the grid in concentric rings (or a similar structure)
would be more natural.
This would, however, entail that the number of grid nodes per layer should
increase at a constant rate, in order to maintain a constant density of nodes
(alongside constant link length, etc.).
However, adding additional nodes permits to distribute skews \emph{better,}
therefore our simplification acts against us.

All other nodes $(x,\ell)$ for $\ell>0$ are TRIX nodes.
Each TRIX node propagates the clock signal to the three nodes ``above'' it,
i.e.,~the vertices $(x - 1, y + 1)$, $(x, y + 1)$, and $(x + 1, y + 1)$.
In the case of the clock generators, the signal is just the generated clock pulse;
in case of the TRIX nodes, this signal is the forwarded clock pulse.
Each node (locally) triggers the pulse, i.e., forwards the signal, when
receiving the second signal from its predecessors;
this way, a single faulty in-neighbor cannot cause the node's pulse to happen
earlier than the first correct in-neighbor's signal arriving or later than the
last such signal.

Pulse propagation over a comparatively long distance involves delays, and our
model focuses on the uncertainty on the wires.
Specifically, we model the wire delays using i.i.d.~random variables that are
fair coin flips, i.e., attain the values~$0$ or~$1$ with probability~$1/2$
each.\footnote{This model choice is restrictive in that it deliberately neglects
correlations.
A partial justification here is the expectation that (positive) correlations
are unlikely to introduce local ``spikes'' in pulse times, i.e., large skews.
However, we acknowledge that this will require further study, which must be
based on realistic models of the resulting physical implementations.}
This reflects that any (known expected) absolute delay does not matter, as the
number of wires is the same for any path from layer~$0$ (the clock generation layer) to layer $\ell>0$; also,
this normalizes the uncertainty to~1.

Formalizing the above, for each wire from the nodes $(x + c, y)$ with $c \in
\{-1, 0, +1\}$ to node $(x,y + 1)$, we define~$w_{c}$ to be the wire delay.
We further define the time $t_{c} \defeq d(x + c, y) + w_{c}$ at which node $(x,
y + 1)$ receives the clock pulse.
Then node $(x, y + 1)$ fires a clock pulse at the median time $t \defeq
\median\{t_{-1}, t_{0}, t_{+1}\}$.
As we assume that all clock generators (i.e., nodes with $y=0$) fire, by
induction on~$y$ all $d(x, y)$ are well-defined and finite.

This model may seem idealistic, especially our choice of the wire delay distribution.
However, in \Cref{sec:systematic} we argue that this is not an issue.

We concentrate on two important metrics to analyze this system:
absolute delay and relative skew.
The total delay $d(x, y)$ of a node $(x, y)$ (usually with $y=H$) is the time at
which this node fires.
The relative skew $s^{\delta}(x, y)$ in horizontal distance $\delta$ is the
difference in total delay; i.e., $s^{\delta}(x, y) \defeq d(x + h, y) - d(x,
y)$.
Our main interests are the random variables $d(H) \defeq d(0,H)$, i.e.~the
delay at the top, and $s(H) \defeq s^{1}(0, H)$, i.e.~the
relative skew between neighboring nodes.

\paragraph*{Further Notation}

We write $\normal(\mu, \sigma^{2})$ to denote the normal
distribution with mean $\mu$ and standard deviation $\sigma$.

Given the average sample value $\bar{v}=\sum_{i=1}^n v_i/n$ over a set of~$n$ sample values~$v_{i}$,
we compute the \emph{empiric} standard deviation as
$(\sum_{i=1}^{n} (v_{i} - \bar{v})^2 / (n - 1))^{1/2}$.

We denote the cumulative distribution function of a random variable~$X$
as~$C[X]$; e.g.~the term $C[X](x)$ denotes the probability that $X \leq x$.  We
denote the inverse function by~$C^{-1}[X]$.

A quantile-quantile-plot relates two distributions $A$ and~$B$.  The
simplest definition is to plot the domain of~$A$ against the domain of~$B$,
using the function $C^{-1}[B](C[A](x))$.

We define the sample space $\Omega$ as the set of possible specific assignments
of wire delays. If we need to refer to the value of some random variable~$X$ in
a sample $s \in \Omega$, we write~$X[s]$.

\section{Delay is Tightly Concentrated}\label{sec:tightdelay}

\subsection{Empiric Analysis}\label{subsec:delay_empiric}

In this subsection, we examine $d(2000)$, which we formally defined in \Cref{subsec:model}.
Recall that $d(2000)$ is the delay at layer~2000.
The results are similar for other layers and, as we will show in
\Cref{subsec:delayasymp}, do change slowly with increasing~$H$.

For reference, consider a simpler system consisting of a sequence of nodes
arranged in a line topology, where each node transmits a pulse once receiving
the pulse from its predecessor.
In this system, the delay at layer~$H$ would follow a binomial distribution with
mean~$H/2=1000$ and standard deviation $\sqrt{H}/2=\sqrt{2000}/2 \approx 22$.
Recall that by the central limit theorem, for large~$H$ this distribution will
be very close to a normal distribution with the same mean and standard
deviation.
In particular, it will have an exponential tail.

\Cref{fig:delay_l2k_pmf} shows the estimated probability mass function
of~$d(2000)$. The data was gathered using 25~million independent
simulations; in \Cref{sec:systematic} we discuss how we ensured that the
simulations are correct.

Observe that the shape looks like a normal distribution, as one might expect. 
However, it is concentrated much more tightly around its mean than for the simple
line topology considered above: The empiric standard deviation of this sample is
only~2.741.

\begin{figure}
  \centering
  \begin{minipage}{.49\textwidth}
    \centering
    \includegraphics[width=\linewidth]{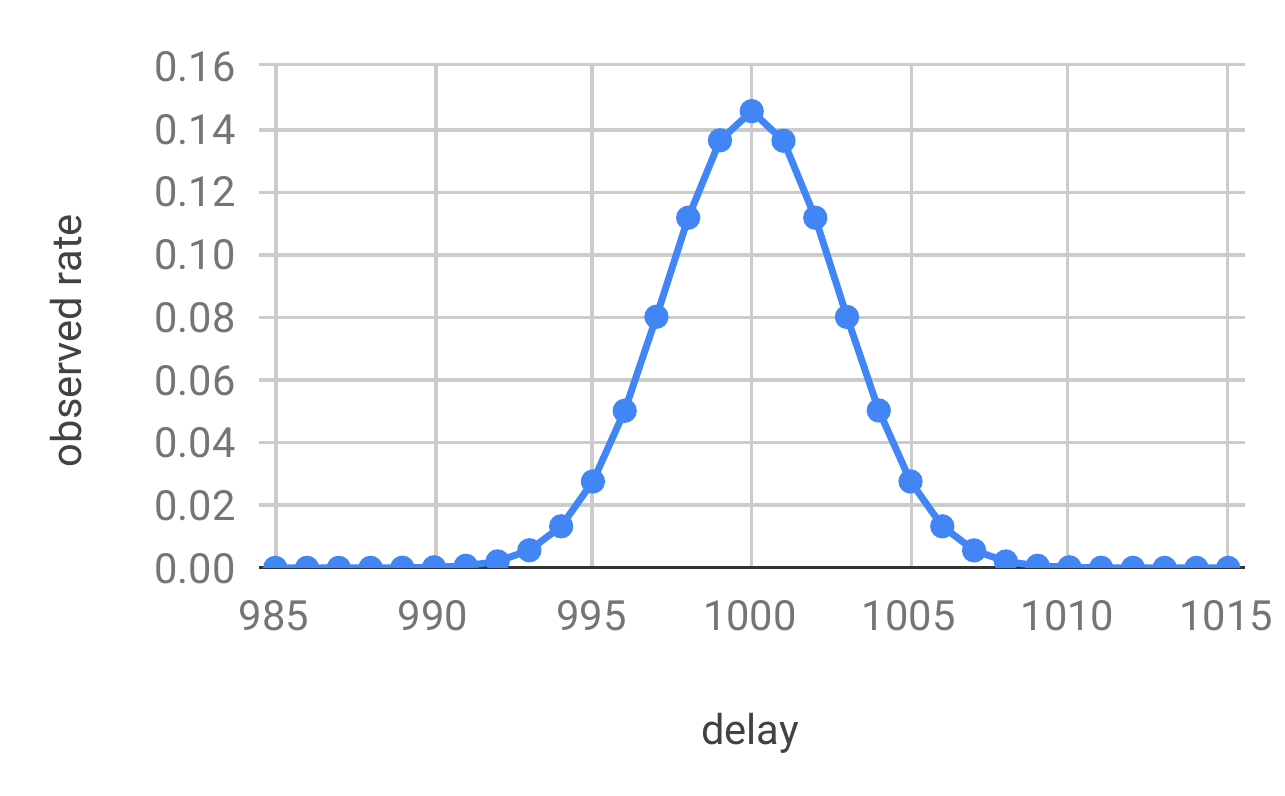}
    \googlevspace{}
    \captionof{figure}{Estimated probability mass function of~$d(2000)$.}
    \label{fig:delay_l2k_pmf}
  \end{minipage}\hfill%
  \begin{minipage}{.49\textwidth}
    \centering
    \includegraphics[width=\linewidth]{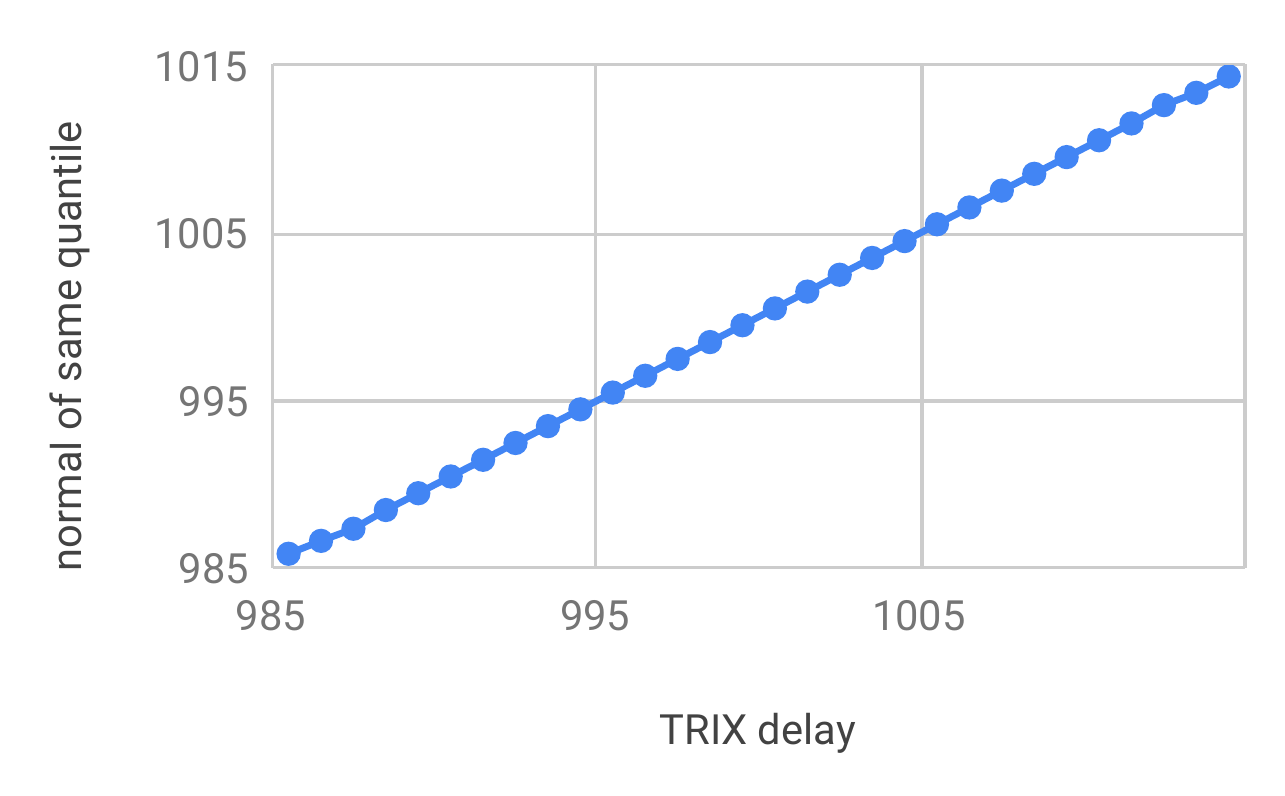}
    \googlevspace{}
    \captionof{figure}{Quantile-quantile plot of~$d(2000)$ against $N(1000,2.741^{2})$.}
    \label{fig:delay_l2k_qq}
  \end{minipage}
\end{figure}

\Cref{fig:delay_l2k_qq} uses a quantile-quantile-plot to show
that $d(2000)$ and $\normal(1000, 2.741^{2})$ seem to be close to
identical, as indicated by the fact that the plot is close to a straight line.
The extremes are an exception, where numerical and
uncertainty issues occur. Of course, they cannot be truly identical, even if our guess~$2.741$ was correct:
$d(2000)$~is discrete and has a bounded support, in contrast to $\normal(1000, 2.741^{2})$.
However, this indicates some kind of connection we would like to understand better.

The Dvoretzky-Kiefer-Wolfowitz inequality~\cite{dvoretzky1956} implies that the true
cumulative distribution function must be within~$0.0003255$ (i.e., $8139$~evaluations)
of our measurement with probability $1 - \alpha = 99\%$.
For the values with low frequency however, the Dvoretzky-Kiefer-Wolfowitz inequality
yields weak error bounds.

Instead, we use that Chernoff bounds can be applied
to the random variable~$X_k$ given by sum of variables~$X_{i,k}$ indicating
whether the $i$-th evaluation of the distribution attains value~$k$.
We then vary~$p_k$, the unknown probability that the underlying distribution
attains value~$k$, and determine the threshold~$p_{\max}$ at which Chernoff's
bound shows that $p_k\geq p_{\max}$ implied that our observed sample had an
a-priori probability of at most $\alpha'$;
the same procedure is used to determine the threshold~$p_{\min}$ for which
$p_k\leq p_{\min}$ would imply that the observed sample has a-priori probability
at most $\alpha'$.
Note that we can also group together multiple values of~$k$ into a single bucket
and apply this approach to the frequency of the overall bucket.
This can be used to address all values with frequency~$0$ together.
Finally, we chose $\alpha'$ suitably such that a union bound over all
buckets\footnote{That is, the number of non-zero values to which we do not apply
Dvoretzky-Kiefer-Wolfowitz plus one (for values of frequency~$0$).}
yields the desired probability bound of $1-\alpha=99\%$ that \emph{all}
frequencies of the underlying distribution are within the computed error bounds.

Due to our large number of samples, the resulting error bars for the probability
mass function are so small that they cannot be meaningfully represented in
\Cref{fig:delay_l2k_pmf}; in fact, on the interval $[990, 1010]$ the error bars
are at most~$8.05\%$ (multiplicative, not additive), and on the
interval $[994,1006]$ at most~$0.93\%$.
On the other hand, data points outside $[990, 1010]$ in \Cref{fig:delay_l2k_qq}
should be taken as rough indication only.

In other words, we have run a sufficient number of simulations to conclude that
the ground truth is likely to be very close to a binomial (i.e., essentially
normal) distribution with mean $H/2=1000$ and standard deviation close to
$2.741$; the former is easily shown, which we do next.

\subsection{Stochastic Analysis}

\begin{lemma}\label{lem:avgdelayishalflayers}
    $\expected{d(H)} = H/2$.
\end{lemma}
\begin{proof}
    Consider the bijection $f\colon \Omega \to \Omega$ on the sample space
    given by $f(s)=\bar{s}$, i.e., we exchange all delays of~$0$ for delays of
    $1$ and vice versa.
    We will show that for a sample~$s$ with $d(2000)[s] = \delta$ it holds
    that $d(2000)[f(s)] = H - \delta$.
    As the wire delays are u.i.d., all points in $\Omega$ have the
    same weight under the probability measure, implying that this is sufficient
    to show that $\expected{d(H)} = H/2$.

    We prove by induction on~$y$ that $d(x,
    y)[f(s)] = y - d(x, y)[s]$ for all $0\leq y\leq H$ and $x\in \Z$;
    by the above discussion, evaluating this claim at $(x,y)=(0,H)$ completes
    the proof.
    For the base case of $y=0$, recall that $d(x, 0)[f(s)] = d(x, 0)[s] = 0$ by
    definition.
    
    For the step from~$y$ to~$y+1$, consider the node at $(x, y + 1)$ for $x\in
    \Z$.
    For $c \in \{-1, 0, +1\}$, the wire delay~$w_c$ from $(x+c,y)$ to $(x,y+1)$
    satisfies $w_{c}[f(s)] = 1 - w_{c}[s]$ by construction.
    By the induction hypothesis, $d(x + c, y)[f(s)] = y - d(x + c, y)[s]$.
    Together, this yields that $(x,y+1)$ receives the pulse from $(x+c,y)$ at
    time
    \begin{align*}
    d(x + c, y)[f(s)] + w_c[f(s)]
    &= y + 1 - (d(x + c, y)[s] + w_c[s]).
    \end{align*}
    We conclude that
    \begin{align*}
    &d(x,y+1)[f(s)]
    =\,\median_{c\in \{-1,0,1\}}\{y + 1 - (d(x + c, y)[s] + w_c[s])\}\\
    &=\, y + 1 - \median_{c\in
    \{-1,0,1\}}\{d(x + c, y)[s] + w_c[s]\}
    = y + 1 - d(x,y+1)[s]. 
    \end{align*}
\end{proof}

\subsection{Asymptotics in Network Depth}\label{subsec:delayasymp}

As discussed earlier, forwarding the pulse signal using a line topology would
result in~$d(H)$ being a binomial distribution with mean~$H/2$ and standard
deviation $\sqrt{H}/2$.
For~$d(2000)$, we observe a standard deviation that is smaller by about an order of
magnitude.

\begin{figure}
  \centering
  \begin{minipage}{\textwidth}
    \centering
    \includegraphics[width=\linewidth]{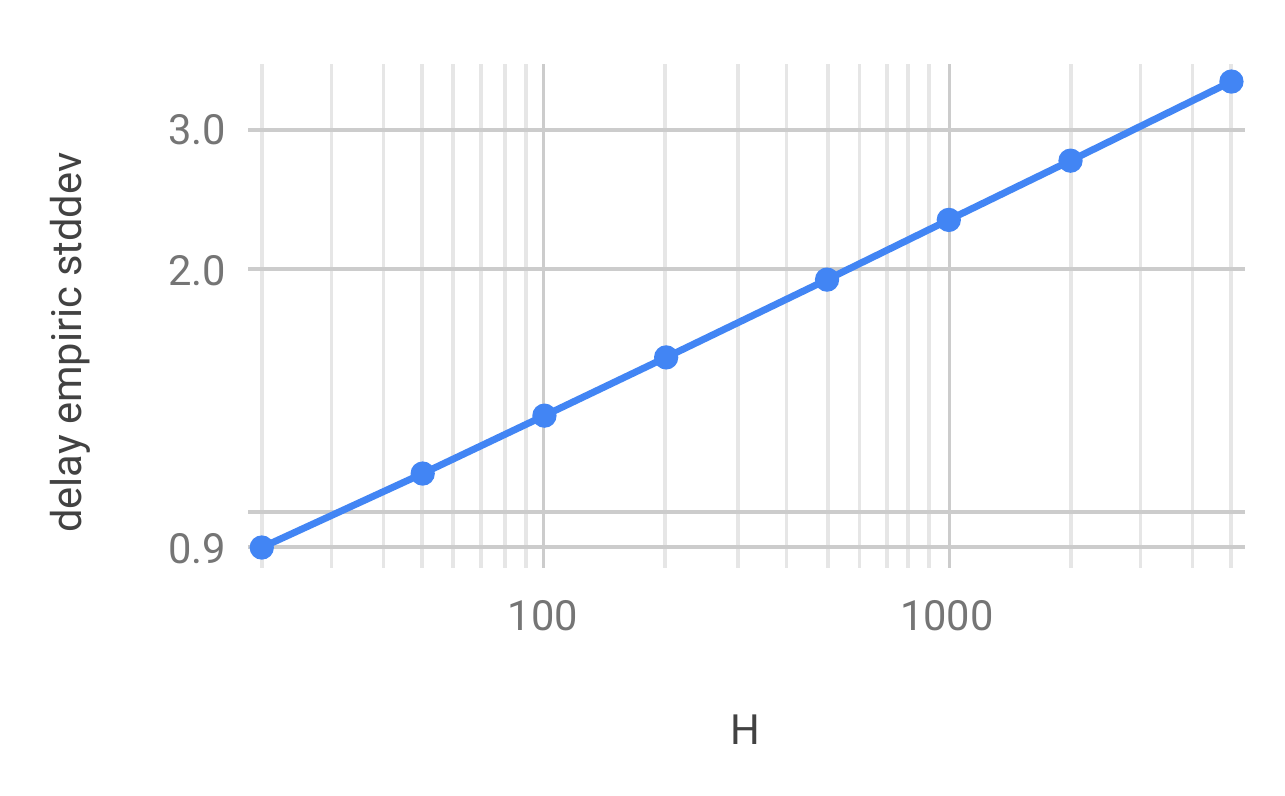}
    \googlevspace{}
    \captionof{figure}{Log-log plot of the empiric standard deviation of~$d(H)$ as a function of~$H$.}
    \label{fig:delay_stddev_asymptotic_loglog}
  \end{minipage}
\end{figure}

Running simulations and computing the empiric standard deviation
for various values of~$H$ resulted in the data plotted in
\Cref{fig:delay_stddev_asymptotic_loglog} as a log-log plot.
Using the technique discussed above, we can compute error bars.\footnote{
This relies on the exponential tails demonstrated in \cref{subsec:delay_empiric}.
Without this additional observation, the error bars for all possible skews (including large skews like $H/2$)
would cause large uncertainty in the standard deviation.}
Again, the obtained error bounds cannot be meaningfully depicted;
errors are about~$1\%$ (multiplicative, not additive) with probability
$1-\alpha=99\%$.
The largest error margin is at~200 layers, with~$1.35\%$.

\Cref{fig:delay_stddev_asymptotic_loglog} suggests a polynomial relationship between
standard deviation $\sigma$ and grid height~$H$.
The slope of the line is close to~$1/4$, which suggests 
$\sigma \sim H^{\beta}$ with $\beta\approx 1/4$.

This is a quadratic improvement over the reference case of a line topology.

\section{Skew is Tightly Concentrated}\label{sec:skewtight}

\subsection{Empiric Analysis}\label{subsec:skew_empiric}

In this subsection, we study $s(2000)$, which we formally defined in \Cref{subsec:model}.
Recall that $s(2000)$ is the skew at layer~2000 between neighboring nodes.
As we will show in \Cref{subsec:skewasymp}, the behavior for other layers is
very similar.
In particular, the skews increase stunningly slowly with~$H$.

We gathered data from 20~million simulations\footnote{Curiously, we saw skew -7 exactly once, skew -6 never, skew +6 four times. Further investigation showed this to be a fluke, but we want to avoid introducing bias by picking the \enquote{nicest} sample.},
and see a high concentration around~$0$ in \Cref{fig:skew_l2k_pmf_linlog}, with
roughly half of the probability mass at~$0$.
Note that again the error bars cannot be represented meaningfully in \Cref{fig:skew_l2k_pmf_linlog};
in fact, on the interval $[-3, +3]$ the error bars are at most~$6.1\%$ (multiplicative, not additive),
and on the interval $[-2, +2]$ at most~$1.4\%$.

Observe that the skew does not follow a normal distribution at all:
The probability mass seems to drop off exponentially like $e^{-\lambda |x|}$ for
$\lambda \approx 2.9$ (where~$x$ is the skew), and not quadratic-exponentially
like $e^{-x^{2}/(2\sigma^2)}$, as it would happen in the normal distribution.
The probability mass for~$0$ is a notable exception, not matching this behavior.

The Dvoretzky-Kiefer-Wolfowitz inequality~\cite{dvoretzky1956} implies that
the true cumulative distribution function must be within~$0.0005147$
(i.e., $5147$~evaluations) of our measurement with probability $1 - \alpha = 99\%$.
In particular, observing skew~$6$ twice without observing skew $-6$ is well
within error tolerance.

\begin{figure}
  \centering
  \begin{minipage}{.48\textwidth}
    \centering
    \includegraphics[width=\linewidth]{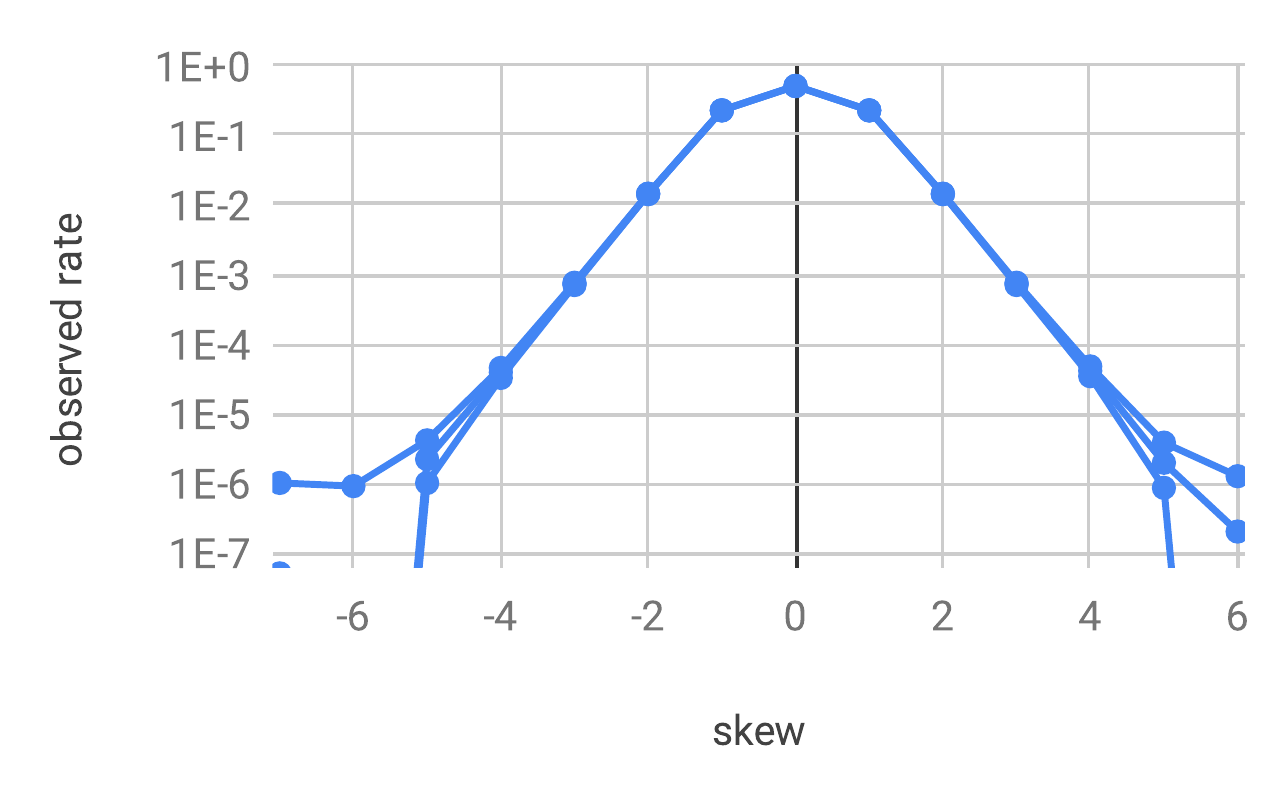}
    \googlevspace{}
    \caption{The estimated probability mass function for~$s(2000)$, with a logarithmic y-axis. The error bounds are only visible at the fringes.}
    \label{fig:skew_l2k_pmf_linlog}
  \end{minipage}\hfill%
  \begin{minipage}{.48\textwidth}
    \centering
    \includegraphics[width=\linewidth]{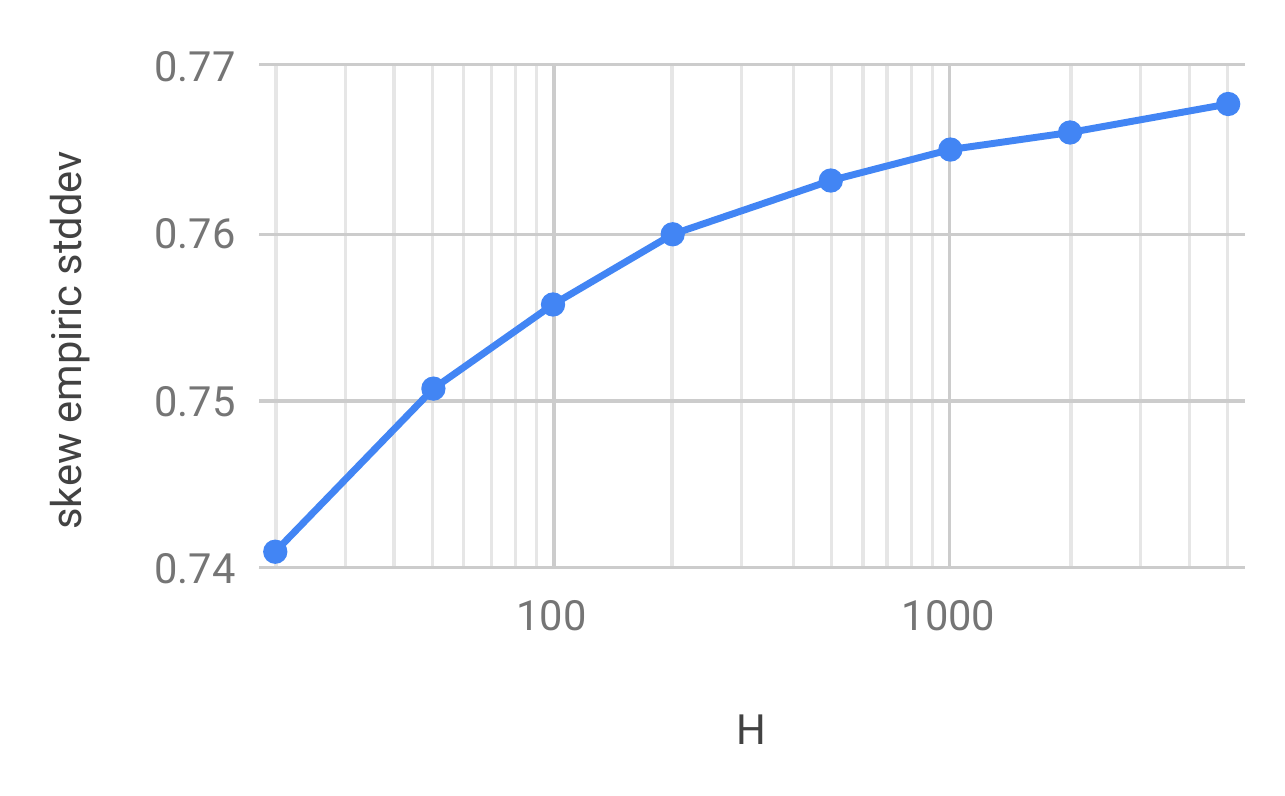}
    \googlevspace{}
    \caption{Empiric standard deviation of~$s(H)$ as a function of~$H$, as a
    log-lin plot.\\~}
    \label{fig:skew_stddev_asymptotic_loglin}
  \end{minipage}\hfill%
\end{figure}

\subsection{Stochastic Analysis}

First, we observe that the skew is symmetric with mean~$0$.
This is to be expected, as the model is symmetric.
It also readily follows using the same argument as used in the proof of
\Cref{lem:avgdelayishalflayers}.
\begin{corollary}\label{cor:symmetric}
$s(H)$ is symmetric with $\expected{s(H)}=0$.
\end{corollary}
\begin{proof}
Consider the bijection used in the proof of
\Cref{lem:avgdelayishalflayers}.
As shown in the proof, for any $(x,y)\in \Z\times \N_0$ and $s\in \Omega$, we
have that
\begin{align*}
d(x,y)[f(s)]-d(x+1,y)[f(s)]
=\,&y-d(x,y)[s]-(y-d(x+1,y)[s])\\
=\,& -(d(x,y)[s]-d(x+1,y)[s]). 
\end{align*}
\end{proof}

Next we prove that the worst-case skew on layer~$H$ is indeed~$H$,
c.f.~\Cref{fig:trix_worst}.

\begin{lemma}\label{lem:extremeskewpossible}
    There is a sample~$s$ such that $s(H)[s] = H$.
\end{lemma}
\begin{proof}
    In~$s$, we simply let all wire delays~$w_{i}$ of wires leading
    to a node with positive~$x$ be~$1$, and let all wire delays~$w_{j}$ of
    wires leading into a node with non-positive~$x$ be~$0$.
    A simple proof by induction shows that for positive~$x$ we get $d(x,
    y)[s_{1}] = y$ and for non-positive~$x$ we get $d(x, y)[s_{1}] = 0$.
    We conclude that $s(H)[s] = d(1, y)[s]-d(0,y)[s] = H$.
\end{proof}

The constructed sample is not the only one which exhibits large skew.
For example, simultaneously changing the delay of all wires between $x=0$ and
$x=1$ does not affect times when nodes pulse.
Moreover, for all $x\notin \{0,1\}$, we can concurrently change the delay of any
one of their incoming wires without effect.
It is not hard to see that the total probability mass of the described
samples is small. We could not find a way to show that this is true in general;
however, the experiments from \Cref{subsec:skew_empiric}
strongly suggests that this is the case. (Hence our question: How to approach statistical problems like this?)

\subsection{Asymptotics in Network Depth}\label{subsec:skewasymp}

Again we lack proper models to describe the asymptotic behavior, and instead
calculate the empiric standard deviation from sufficiently many simulations.

\begin{figure}
  \centering
  \begin{minipage}{.48\textwidth}
    \centering
    \includegraphics[width=\linewidth]{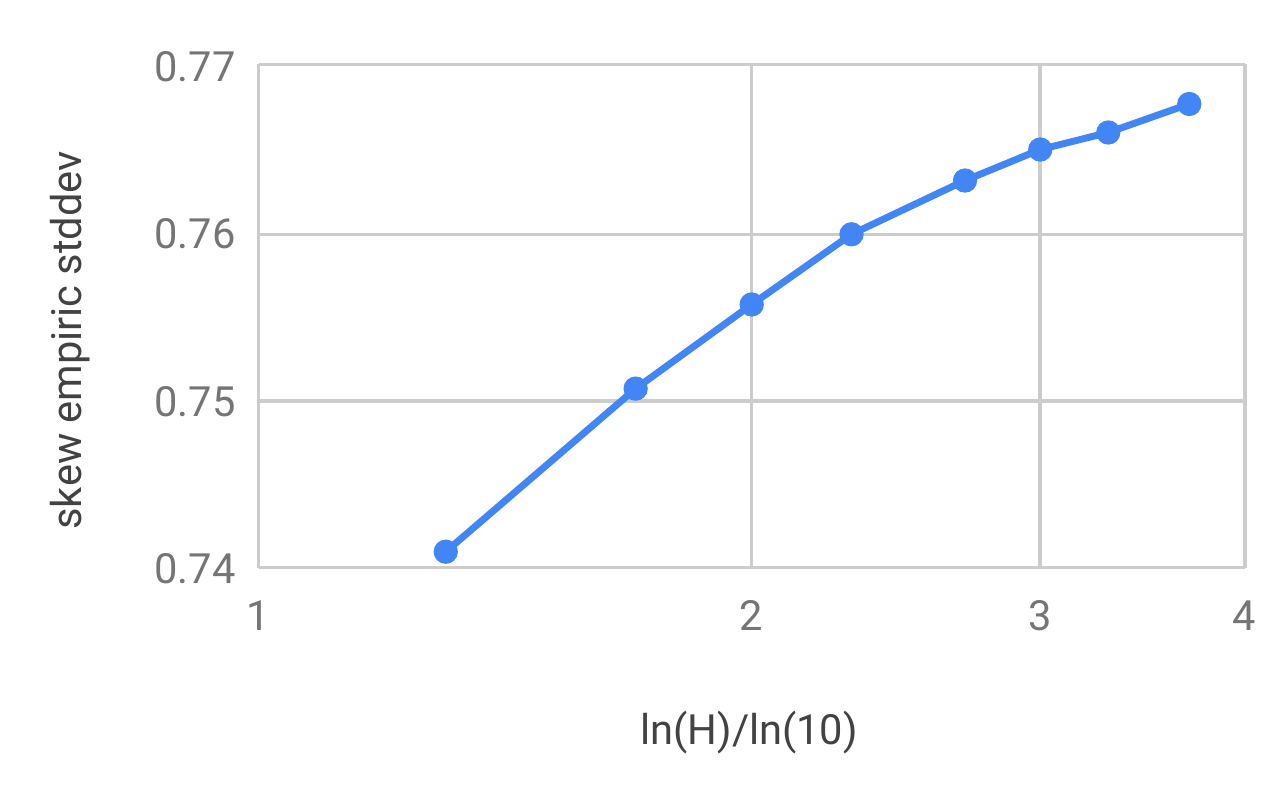}
    \googlevspace{}
    \caption{The data of \cref{fig:skew_stddev_asymptotic_loglin} in a loglog-lin plot.\\~}
    \label{fig:skew_stddev_asymptotic_logloglin}
  \end{minipage}%
  \hfill{}%
  \begin{minipage}{.48\textwidth}
    \centering
    \includegraphics[width=\linewidth]{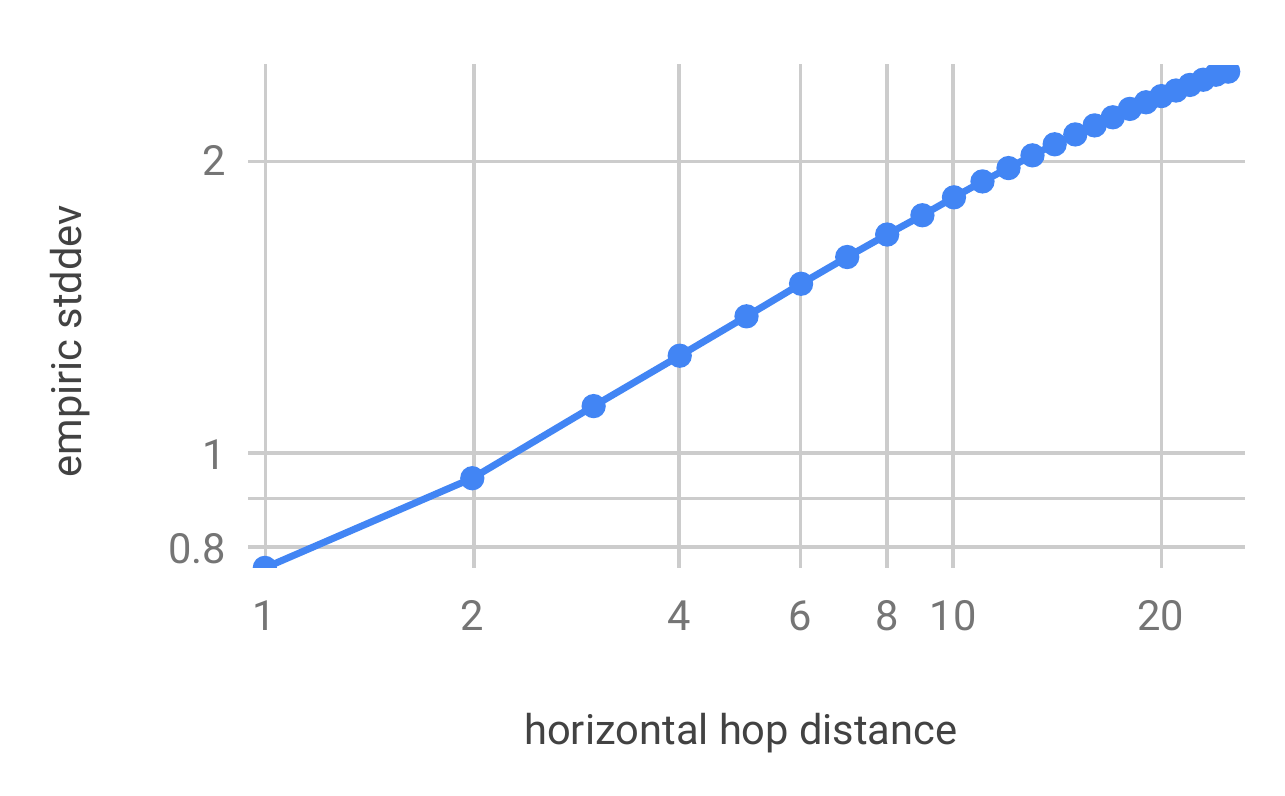}
    \googlevspace{}
    \caption{Empiric standard deviation of $s^{\delta}(0,500)$ as a function of
    horizontal distance~$\delta$ in a log-log plot.}
    \label{fig:skew_stddev_dlin_loglog}
  \end{minipage}
\end{figure}

\Cref{fig:skew_stddev_asymptotic_loglin} shows that the skew remains small even
for large values of~$H$.\footnote{A standard deviation below $1$~link delay uncertainty is not
unreasonable: Two circuits clocked over independent links, hop-distance~$1$ from a perfect clock
source, have $\sigma(skew)=\sqrt{\mathbb{E}[0.5^{2}]} = 0.5$ uncertainties.}
Note that the X-axis is logarithmic.
Yet again, showing the error bars calculated using Chernoff bounds and DKW (Dvoretzky-Kiefer-Wolfowitz) is
not helpful, as we get relative errors of about $0.6\%$ with $99\%$ confidence.
The largest error margin is at~100 layers and below, with an error about~$0.95\%$.

This suggests that the standard-deviation of~$s(H)$ grows strongly
sub-loga\-rithmi\-cally, possibly even converges to a finite value.
In fact, the data indicates a growth that is significantly slower than
logarithmic:
in \Cref{fig:skew_stddev_asymptotic_logloglin}, the x-axis is doubly logarithmic.
Thus, the plot suggests that $s(H)\in O(\log \log H)$.

Note that if we pretended that adjacent nodes exhibit independent delays,
the skew would have the same concentration as the delay.
In contrast, we see that adjacent nodes are tightly synchronized,
implying strong correlation.\footnote{For the sake of brevity, we do not show
respective plots.}

\subsection{Asymptotics in Horizontal Distance}\label{subsec:skewdlin}

So far, we have limited our attention to the skew between neighboring nodes.
In contrast, at horizontal distances $\delta\geq 2H$, node delays
are independent, as they do not share any wires on any path to any clock
generator.
Therefore, the skew would be given by independently sampling twice from the
delay distribution and taking the difference.
It remains to consider how skews develop for smaller horizontal distances.

In \Cref{fig:skew_stddev_dlin_loglog}, we see that the skew grows steadily with
increasing $\delta$.
These simulations were run with $H=500$ for higher precision and execution speed.
Note that the standard deviation of the skew converges to roughly $\sqrt{2}\cdot
1.95\approx 2.75$, i.e., $\sqrt{2}$ times the standard deviation of the delay
at $H=500$, as $\delta$ approaches $2\cdot 500 = 1000$.
This holds because the sum of independent random variables has variance equal to the
sum of variances of the variables.
As before, depicting the error bars determined by our approach is of no use,
as relative errors are about~$1.6\%$.
The largest error margin is at horizontal distance~25, with~$1.79\%$.

The small relative errors indicate that the log-log representation of the data
is meaningful, and suggests that
the standard deviation increases roughly proportional to
$\delta^{\gamma}$ for $\gamma \approx 1/3$.
It is not surprising that the slope falls off towards larger values, as we know
that the curve must eventually flatten and become constant for $\delta \ge 2H
= 1000$.
Overall, we observe that the correlation of skew at a distance is stronger than expected,
specifically $\gamma \approx 1/3$ instead of the expected $\gamma \approx 1/2$ for small
$\delta$.

\section{Conclusion}\label{sec:conclude}

In this work, we studied the behavior of the TRIX grid under u.i.d.\ link delays,
using statistical tools.
Our results clearly demonstrate that the TRIX grid performs much better than
one would expect from naive solutions, and thus should be considered
when selecting a fault-tolerant clock propagation mechanism.

Concretely, our simulation experiments show that the delay as function of the distance~$H$ from the clock source
layer is close to normally distributed with a standard deviation that grows only
as roughly~$H^{1/4}$, a qualitative change from e.g.\ a line topology that would
not be achieved by averaging or similar techniques.
Moreover, the skew between neighbors in the same layer is astonishingly small.
While not normally distributed, there is strong evidence for an exponential
tail, and the standard deviation of the distribution as function of~$H$ appears
to grow as $O(\log \log H)$.
Checking the skew over larger horizontal counts~$d$ (within the same layer) shows a
less surprising pattern.
However, still the increase as function of~$d$ appears to be slightly slower
than $\sqrt{d}$.

These properties render the TRIX grid an attractive candidate for
fault-tolerant clock propagation, especially when compared to clock trees.
We argue that our results motivate further investigation, considering correlated
delays based on measurements from physical systems as well as simulation of
frequency and/or non-white phase noise.
In addition, the impact of faults, physical realization as less regular grid
with a central clock source, and the imperfection of the input provided by the
clock source need to be studied.

Last but not least, we would like to draw attention to the open
problem of analyzing the stochastic process we use as an abstraction for
TRIX. This is also the reason why we lack a purely mathematical analysis.
While our simulation experiments are sufficient to demonstrate that the exhibited behavior is
highly promising, gaining an understanding of the underlying cause would allow
making qualitative and quantitative predictions beyond the considered setting.
As both the nodes' decision rule and the topology are extremely simple, one may
hope for a general principle to emerge that can also be applied in different
domains.

\clearpage

\appendix
\section{Potential Systematic Errors}\label{sec:systematic}

In this section we discuss possible sources of systematic errors in our
simulations and how we guarded against them.

\subsection{Bugs}

Since all experiments are software simulations, measurements have to be taken to insure against bugs.
In this regard, there are several arguments to be made, some involving the Random Number Generator (RNG):

\begin{itemize}
    \item We cross-validated four different implementations: (1) A very simple
    Python implementation that uses system randomness; (2) a slightly more
    involved Python implementation that exhaustively enumerates all possible
    wire delay combinations; (3) a straight-forward C implementation using system
    randomness; and (4) an optimized C implementation with the slightly weaker
    RNG \enquote{\texttt{xoshiro512starstar}}~\cite{DBLP:journals/corr/abs-1805-01407}.
    \item Even though only implementation (4) was fast enough to be used to
    generate the bulk of the results, all implementations agree on the probability
    distributions of delay and skew for examples they can handle. Implementation (2) can
    only run up to layer~3 ($H \leq 3$), implementations (1) and (3) were used up to
    around layers~100 and~1000, respectively.
    \item Implementation (4) is short (200 lines, plus about
    100 lines for the RNG~\cite{DBLP:journals/corr/abs-1805-01407}), and is simple enough
    to be inspected manually.
    \item Multiple machines were used, so hardware failure can be ruled out with
    sufficient confidence.
\end{itemize}

\subsection{Randomness and Model}

\begin{itemize}
    \item Tests with a number of weak RNGs showed that TRIX seems
    to be robust against this kind of deviation.
    \item Explorative simulations and validation simulations were kept strictly separate.
    \item For most discussions we assumed $H=2000$, because this definitely covers all practical applications.
    In fact, we expect that many applications only need $H=200$ or even $H=20$.
    In this paper we use a large value for~$H$ to show that the observed
    behavior is not a fluke that occurs due to low~$H$, but that the growth of
    delay and skew as function of~$H$ is indeed asymptotically slow.
    \item Using a larger domain only scales the result linearly, as expected.
    \item Using different wire delay models (e.g.\ choosing uniformly from
    $\{0, 0.5, 1\}$ instead of $\{0, 1\}$) does not significantly
    change the result, and in fact improves it slightly, as expected.
    \item We only focused on fault-free executions. Observe that single isolated
    faults only introduce an additional uncertainty of at most~1 (recall the
    normalization~$u = 1$). As faults are (supposed to be) rare and not
    maliciously placed, this means that the predictions for fault-free systems
    have substantial and meaningful implications also for systems with faults.
\end{itemize}

\section{Raw Data}

\noindent In the following, we provide the data for all graphs.

\newcommand{\appdelaypmfqq}{
\begin{center}
\begin{tabular}{rr||rr}
\multicolumn{2}{c||}{\Cref{fig:delay_l2k_pmf}} & \multicolumn{2}{c}{\Cref{fig:delay_l2k_qq}} \\
x (delay) & y (rate) & x (delay) & y (normal) \\ \hline
985 & 0.00000012 & 985.5 & 985.843 \\
986 & 0.00000040 & 986.5 & 986.615 \\
987 & 0.00000140 & 987.5 & 987.338 \\
988 & 0.00001076 & 988.5 & 988.457 \\
989 & 0.00004740 & 989.5 & 989.460 \\
990 & 0.00019060 & 990.5 & 990.462 \\
991 & 0.00066280 & 991.5 & 991.457 \\
992 & 0.00206788 & 992.5 & 992.464 \\
993 & 0.00561240 & 993.5 & 993.470 \\
994 & 0.01324480 & 994.5 & 994.472 \\
995 & 0.02753772 & 995.5 & 995.475 \\
996 & 0.05006436 & 996.5 & 996.479 \\
997 & 0.08002964 & 997.5 & 997.486 \\
998 & 0.11158960 & 998.5 & 998.492 \\
999 & 0.13625144 & 999.5 & 999.498 \\
1000 & 0.14553656 & 1000.5 & 1000.503 \\
1001 & 0.13611860 & 1001.5 & 1001.508 \\
1002 & 0.11158608 & 1002.5 & 1002.515 \\
1003 & 0.07996644 & 1003.5 & 1003.520 \\
1004 & 0.05016040 & 1004.5 & 1004.526 \\
1005 & 0.02753824 & 1005.5 & 1005.531 \\
1006 & 0.01324832 & 1006.5 & 1006.537 \\
1007 & 0.00557100 & 1007.5 & 1007.542 \\
1008 & 0.00205200 & 1008.5 & 1008.545 \\
1009 & 0.00066300 & 1009.5 & 1009.545 \\
1010 & 0.00018904 & 1010.5 & 1010.552 \\
1011 & 0.00004660 & 1011.5 & 1011.556 \\
1012 & 0.00001044 & 1012.5 & 1012.650 \\
1013 & 0.00000144 & 1013.5 & 1013.385 \\
1014 & 0.00000044 & 1014.5 & 1014.363 \\
1015 & 0.00000008 &  &  \\
\end{tabular}
\end{center}
}

\newcommand{\appdelayskewasymptotic}{
\begin{center}
\begin{tabular}{r|r|r}
& \Cref{fig:delay_stddev_asymptotic_loglog} & \Cref{fig:skew_stddev_asymptotic_loglin,fig:skew_stddev_asymptotic_logloglin} \\ 
x ($H$) & y (stddev delay) & y (stddev skew) \\ \hline
20 & 0.9012352 & 0.74096549 \\
50 & 1.1147817 & 0.75071102 \\
100 & 1.3166986 & 0.75573837 \\
200 & 1.5567596 & 0.75993333 \\
500 & 1.9468302 & 0.76314438 \\
1000 & 2.3115292 & 0.76499614 \\
2000 & 2.7406959 & 0.76601516 \\
5000 & 3.4405614 & 0.76771794 \\
\end{tabular}
\end{center}
}

\newcommand{\appskewpmf}{
\begin{center}
\begin{tabular}{rccc}
\multicolumn{4}{c}{\Cref{fig:skew_l2k_pmf_linlog}} \\ 
x (skew) & lowerb. rate & observed rate & upperb. rate \\ \hline
-7 & 0.0000000 & 0.0000001 & 0.0000010 \\
-6 & 0.0000000 & 0.0000000 & 0.0000009 \\ 
-5 & 0.0000010 & 0.0000022 & 0.0000041 \\
-4 & 0.0000324 & 0.0000389 & 0.0000452 \\
-3 & 0.0007020 & 0.0007322 & 0.0007578 \\
-2 & 0.0142418 & 0.0143777 & 0.0144897 \\
-1 & 0.2283952 & 0.2287592 & 0.2291231 \\
0 & 0.5120736 & 0.5124375 & 0.5128015 \\
1 & 0.2281477 & 0.2285116 & 0.2288756 \\
2 & 0.0142291 & 0.0143650 & 0.0144769 \\
3 & 0.0007024 & 0.0007326 & 0.0007583 \\
4 & 0.0000343 & 0.0000410 & 0.0000475 \\
5 & 0.0000009 & 0.0000020 & 0.0000038 \\
6 & 0.0000000 & 0.0000002 & 0.0000013 \\
\end{tabular}
\end{center}
}

\newcommand{\appskewdlin}{
\begin{center}
\begin{tabular}{rc}
\multicolumn{2}{c}{\Cref{fig:skew_stddev_dlin_loglog}} \\ 
hop dist. & empiric stddev \\
1 & 0.76334 \\
2 & 0.94346 \\
3 & 1.11873 \\
4 & 1.26012 \\
5 & 1.38309 \\
6 & 1.49367 \\
7 & 1.59118 \\
8 & 1.67803 \\
9 & 1.75621 \\
10 & 1.83219 \\
11 & 1.90233 \\
12 & 1.96331 \\
13 & 2.02323 \\
14 & 2.07683 \\
15 & 2.12580 \\
16 & 2.17155 \\
17 & 2.21355 \\
18 & 2.25797 \\
19 & 2.29356 \\
20 & 2.32732 \\
21 & 2.35812 \\
22 & 2.38959 \\
23 & 2.41912 \\
24 & 2.44794 \\
25 & 2.46569 \\
\end{tabular}
\end{center}
}


~\\

\noindent
\begin{minipage}{.59\textwidth}
    \appdelaypmfqq{}
\end{minipage}%
\begin{minipage}{.39\textwidth}
    \appskewdlin{}
\end{minipage}

\clearpage

\appdelayskewasymptotic{}
\appskewpmf{}

\bibliographystyle{splncs04}
\bibliography{references}

\end{document}